\documentclass[journal]{IEEEtran}
\usepackage{algorithm}
\usepackage[noend]{algpseudocode} 
\makeatletter 
\def\BState{\State\hskip-\ALG@thistlm}

\makeatother 
\usepackage{authblk}
\usepackage{graphicx}
\usepackage{float}
\usepackage{epsfig}
\usepackage{epstopdf}
\usepackage[caption=false,font=footnotesize]{subfig}
\usepackage{colortbl}
\usepackage[cmex10]{amsmath}
\usepackage{array}
\usepackage{cite}
\usepackage[utf8x]{inputenc}
\usepackage{amssymb}
\usepackage{stackrel}
\usepackage{booktabs,multirow}
\usepackage{etoolbox}
\makeatletter
\patchcmd{\@makecaption}
{\scshape}
{}
{}
{}
\makeatother
\usepackage{amsthm}
\usepackage{amsmath,mathtools}

\usepackage{adjustbox}

\newtheorem{Proposition}{Proposition}

\newtheorem{Example}{Example}
\usepackage[T1]{fontenc}
\usepackage{cuted}
\usepackage{textcomp}
\usepackage{comment}
\usepackage{color}

\begin{document}

\title{Efficient Search of Compact QC-LDPC and SC-LDPC Convolutional Codes with Large Girth}

\author{Mohammad H. Tadayon, Alireza Tasdighi, Massimo Battaglioni, Marco Baldi and Franco Chiaraluce
\thanks{Mohammad H. Tadayon is with Iran Telecommunication Research Center (ITRC), Tehran, Iran
(e-mail: tadayon@itrc.ac.ir).
Alireza Tasdighi is with the Department of Mathematics and Computer Science, Amirkabir University of Technology, Tehran, Iran
(e-mail: a.tasdighi@aut.ac.ir).
Massimo Battaglioni, Marco Baldi and Franco Chiaraluce are with the Dipartimento di Ingegneria dell'Informazione, Universit\`a Politecnica delle Marche, Ancona, Italy (e-mail: m.battaglioni@pm.univpm.it, m.baldi@univpm.it, f.chiaraluce@univpm.it)}
}


\renewcommand\Authands{ and }
\maketitle

\begin{abstract}
We propose a low-complexity method to find quasi-cyclic low-density parity-check block codes with girth $10$ or $12$ and shorter length than those designed through classical approaches. The method is extended to time-invariant spatially coupled low-density parity-check convolutional codes, permitting to achieve small syndrome former constraint lengths. Several numerical examples are given to show its effectiveness.
\end{abstract}


\section{Introduction}
The error rate performance of quasi-cyclic low-density parity-check (QC-LDPC) block codes decoded through iterative algorithms is adversely affected by the presence of cycles with short length in their associated Tanner graphs \cite{Tanner1981}. 
Therefore, the minimum length of cycles, also known as girth of the graph (and denoted by $g$ afterwards), should be kept as large as possible \cite{Wang2008}.
At the same time, short block codes are required in modern applications, like machine-to-machine communications, where low latency must be achieved.

Motivated by these considerations, we propose a method for designing QC-LDPC block codes with large girth ($g = 10, 12$) and short length.
QC-LDPC block codes are the basis for the design of spatially coupled low-density parity-check convolutional codes (SC-LDPC-CCs). These codes can exploit shift register-based circuits for encoding \cite{Felt1}, whereas sliding window (SW) iterative algorithms based on belief propagation (BP) can be used for their decoding \cite{Felt1,Lentmaier2005}.
These decoders perform BP over a window including $W$ blocks of $a$ bits each and, at any decoding stage, give in output the first $a$ decoded bits; the window is then shifted forward by $a$ bits.
The smallest number of blocks required for achieving good performance is $W=\alpha(m_h+1)$, where $m_h$ is the syndrome former memory order of the code and $\alpha \in \mathbb{N}$ usually takes values in $[5,\ldots,10]$.
The decoding latency ($\Lambda_{\mathrm{SW}}$) and per output bit complexity ($\Gamma_{\mathrm{SW}}$) of a SW decoder can be expressed as
\begin{equation}
\begin{cases}
\Lambda_{\mathrm{SW}} = Wa = \alpha(m_h+1)a,\\
\Gamma_{\mathrm{SW}} = \frac{WaI_{\mathrm{avg}} \cdot f(w_{\mathrm{avg}},R)}{a}=\alpha(m_h+1)I_{\mathrm{avg}} \cdot f(w_{\mathrm{avg}},R),
\end{cases}
\label{cas}
\end{equation}
where $I_{\mathrm{avg}}$ is the average number of decoding iterations and $f(w_{\mathrm{avg}},R)$ is a function of the average column weight $w_{\mathrm{avg}}$ of the parity-check matrix and the code rate $R$ (see \cite{Hu2001} for further details). As mentioned, for fixed values of $w_{\mathrm{avg}}$ and $R$, a SW decoder returns $a$ bits per window, independently of the window size $W$. Hence, if $m_h$ is kept small, SC-LDPC-CCs can be decoded with a small window size, thus yielding, according to \eqref{cas}, low decoding latency and complexity. 

We aim at designing QC-LDPC block codes with smaller blocklength and SC-LDPC-CCs with smaller $m_h$ than those with comparable girth available in the literature, and we introduce the notion of \textit{compact codes} to encompass block and convolutional LDPC codes with these features in one word.
In order to design compact codes, we resort to a construction exploiting sequentially multiplied columns (SMCs). The latter are obtained starting from a base column and consecutively multiplying it by suitably chosen coefficients. These columns are used to form the parity-check matrix of a QC-LDPC code (SC-LDPC-CC) and then a greedy search algorithm is used to find compact codes. The SMC assumption significantly reduces the search space and permits us to exhaustively explore it. Moreover, we use a recently introduced integer programming optimization model, called min-max \cite{MBAT2017}, in order to find the minimum possible $m_h$. This model, instead of performing an exhaustive search, takes benefit of a heuristic optimization approach, thus further reducing the search time.

\section{Notation}
\label{Sec2}

\subsection{CPM-based QC-LDPC block codes}

Let us consider a QC-LDPC block code, in which the parity-check matrix is an $m \times n$ array of $N \times N$ circulant permutation matrices (CPMs), $\mathbf{I}(p_{ij})$, $0 \leq i \leq m - 1$, $0 \leq j \leq n - 1$, where $N$ is the \textit{lifting degree} of the code. $\mathbf{I}(p_{ij})$ is obtained from the identity matrix through a cyclic shift of its rows by $p_{ij}$ positions, with $0 \leq p_{ij} \leq N - 1$. The code length is $L=nN$. The $m \times n$ matrix $\mathbf{P}$ having the integer values $p_{ij}$ as its entries is referred to as the {\em exponent matrix} of the code.
For such a QC-LDPC block code, a necessary and sufficient condition for the existence of a cycle of length $2k$ 
in its Tanner graph is \cite{MFossorier1}
\begin{equation}
\sum_{i=0}^{k-1} \left( p_{m_{i}n_{i}} - p_{m_{i}n_{i+1}} \right) = 0  \mod N , 
\label{fore}
\end{equation}   
where $n_{k}=n_{0}$, $m_{i} \neq m_{i+1}$, $n_{i} \neq n_{i+1}$.

To achieve a certain girth $g$, for given values of $m$ and $n$, and for a fixed value of $N$, one has to find a matrix $\mathbf{P}$ whose entries do not satisfy (\ref{fore}) for any value of $k < g/2$, and any possible choice of the row and column indexes $m_i$ and $n_i$. Starting from $\mathbf{P}$, the Tanner graph of the code can be easily obtained, as it is unambiguously related to the values of $p_{ij}$.

We define an {\em avoidable cycle} in the Tanner graph of a CPM-based QC-LDPC block code as a cycle for which $\sum_{i=0}^{k-1} \left( p_{m_{i}n_{i}} - p_{m_{i}n_{i+1}} \right) = \beta N$, $\beta > 0$. A {\em strictly avoidable cycle} is defined as a cycle for which $\sum_{i=0}^{k-1} \left( p_{m_{i}n_{i}} - p_{m_{i}n_{i+1}} \right)=0$. 

\subsection{SC-LDPC convolutional codes}
Time-invariant SC-LDPC-CCs are characterized by semi-infinite parity-check matrices in the form
\begin{equation}
\mathbf{H} = \left[\begin{array}{cccccc}
\arraycolsep=1.4pt\def\arraystretch{5pt}
\mathbf{H}_0 				& \mathbf{0} 							& \mathbf{0} 							& \ddots \\
\mathbf{H}_1				& \mathbf{H}_0					& \mathbf{0} 							& \ddots \\
\mathbf{H}_2				& \mathbf{H}_1					& \mathbf{H}_0					& \ddots \\
\vdots 			& \mathbf{H}_2 				& \mathbf{H}_1 				& \ddots \\
\mathbf{H}_{m_h}	& \vdots 				& \mathbf{H}_2					& \ddots \\
\mathbf{0} 						& \mathbf{H}_{m_h}	& \vdots					& \ddots \\
\mathbf{0} 						& \mathbf{0} 							& \mathbf{H}_{m_h}	& \ddots \\
\vdots				& \vdots					& \vdots					& \ddots \\
\end{array}\right],
\label{eq:Hconv}
\end{equation}
where each block $\mathbf{H}_i$, $i = 0, 1, 2, \ldots, m_h$, is a binary matrix with size $c \times a$.
The syndrome former matrix $\mathbf{H_s} = \left[ \mathbf{H}_0^T | \mathbf{H}_1^T | \mathbf{H}_2^T | \ldots | \mathbf{H}_{m_h}^T \right]$, where $^T$ denotes transposition, has $a$ rows and $(m_h+1)c$ columns. The code defined by the parity-check matrix \eqref{eq:Hconv} has asymptotic code rate $R = \frac{a-c}{a}$;
$m_h$ determines the height of the non-zero diagonal band in \eqref{eq:Hconv}, whereas the code syndrome former constraint length $v_s = (m_h + 1) a$ gives its length.

\subsection{Link between QC-LDPC block codes and SC-LDPC-CCs}
A common representation of the syndrome former matrix $\mathbf{H_s}$ of an SC-LDPC-CC is based on polynomials $\in F_2[x]$, the ring of polynomials with coefficients in the Galois field $F_2$. In this case, the code is described by a $c \times a$ {\em symbolic matrix} with polynomial entries, that is
\begin{equation}
\mathbf{H}(x)=\left[\begin{array}{llll}
h_{0,0}(x) & h_{0,1}(x) & \ldots & h_{0,a-1}(x)\\
h_{1,0}(x) & h_{1,1}(x) & \ldots & h_{1,a-1}(x)\\
\vdots & \vdots & \ddots & \vdots\\
h_{c-1,0}(x) & h_{c-1,1}(x) & \ldots & h_{c-1,a-1}(x)\end{array}\right],
\label{eq:Hx}
\end{equation}
where each $h_{i,j}(x)$, $i = 0, 1, 2, \ldots, c-1$, $j = 0, 1, 2, \ldots, a-1$, is a polynomial $\in F_2[x]$.
The code representation based on $\mathbf{H_s}$ can be converted into that based on $\mathbf{H}(x)$ as follows
\begin{equation}
h_{i,j}(x)=\sum_{m=0}^{m_h} h_{m}^{(i,j)} x^{m},
\label{eq:bintopol}
\end{equation}
where $h_{m}^{(i,j)}$ is the $(i, j)$-th entry of the matrix $\mathbf{H}_m$, the latter being the $m$-th block of $\mathbf{H_s}$.

We focus on codes described by a symbolic parity-check matrix containing only monomials, also known as \textit{monomial codes}. In this case, $\mathbf{H}(x)$ can be described through an exponent matrix in the form 
\begin{equation}
\mathbf{P}=\left[\begin{array}{llll}
p_{0,0} & p_{0,1} & \ldots & p_{0,a-1}\\
p_{1,0} & p_{1,1} & \ldots & p_{1,a-1}\\
\vdots & \vdots & \ddots & \vdots\\
p_{c-1,0} & p_{c-1,1} & \ldots & p_{c-1,a-1}\end{array}\right],
\label{eq:expomatrix}
\end{equation}
where $p_{i,j}$ is the exponent of the (only) non-null term in $h_{i,j}(x)$. This extends the definition of exponent matrix introduced above for a QC-LDPC block code, thus evidencing the link between the two representations.
The syndrome former memory order $m_h$ is the largest difference, in absolute value, between any two elements of $\mathbf{P}$.

\section{Code design via sequentially multiplied columns}
\label{Sec3}

It is shown in \cite{ATasdighi1} that the complexity of exhaustively checking equations of the type (\ref{fore}) goes exponentially high by increasing each one of the parameters $m$, $n$ or $N$.
Solutions with reduced complexity where proposed, for example, in \cite{ATasdighi2,Gholami2017}, but the corresponding design methods result in $g = 8$. Next, we prove that by using the SMC assumption we can instead design codes with girth up to $12$.

Let $m<n\leq N$ ($m,n,N\in \mathbb{N}$) and consider the exponent matrix $\mathbf{P}$ for a QC-LDPC code in the form (SMC assumption)
\begin{equation}
\mathbf{P}^{\mathrm{SMC}}_{m\times n}=\left[\begin{array}{c|c|c|c|c|c}
\vec{0} & \vec{P}_{1} & \gamma_{2} \otimes \vec{P}_{1} & \gamma_{3} \otimes \vec{P}_{1} & \ldots &  \gamma_{n-1} \otimes \vec{P}_{1}
\end{array}\right],
\label{eq:SMCexpomatrix}
\end{equation}
where $\vec{0}$ and $\vec{P}_{1}$ are column vectors of size $m$. $\vec{0}$ is an all zero vector and $\vec{P}_{1}$ is a vector with first (i.e., top most) entry equal to zero, second entry equal to one, while its remaining entries are selected from $\{2,\ldots ,N - 1\}$, in an increasing order. Vectors $\gamma_{j} \otimes \vec{P}_{1}$ ($j = 2,\ldots, n - 1;\; \gamma_{j}\in \{2,\ldots ,N - 1\}$ and $\gamma_{j}<\gamma_{j + 1}$) are obtained from the base vector $\vec{P}_{1}$ through sequential multiplications, and $\otimes$ represents multiplication mod $N$. Let us denote by $\mathcal{I}_{0,1,2,\ldots,j}^{k}$, $j=2,\ldots,n-1$ a set containing all relations (\ref{fore}) corresponding to the potential cycles having lengths ranging between $4$ and $2k$ ($k = 2, 3, 4, 5$) in the Tanner graph of \eqref{eq:SMCexpomatrix}. 
The following proposition holds.
\begin{Proposition}
\label{SMCproposition}
Let
$\mathbf{P}^{\mathrm{SMC}}_{m\times n}$ be the exponent matrix of a QC-LDPC block code $C$ as defined in (\ref{eq:SMCexpomatrix}). Suppose that the Tanner graph associated to the submatrix $\left[\begin{array}{@{}c@{}|@{}c@{}} \vec{0} & \vec{P}_{1} \end{array}\right]$ contains no strictly avoidable cycles of length up to $10$. Then, the Tanner graph of $C$ has no strictly avoidable cycle of length up to $10$ for sufficiently large $N$ and a proper choice of $\gamma_j$'s.
\end{Proposition}
\begin{proof}
Demonstration is conducted inductively, which means that $\gamma_2$ is determined first, followed sequentially by $\gamma_3, \gamma_4, \ldots, \gamma_{n - 1}$.
Since we are considering cycles with length up to $10$, for each element in $\mathcal{I}_{0,1,\ldots,s - 1}^{k}$ we can write a relation of type \eqref{fore} consisting of, at most, five parts (depending on the overall length of a considered cycle), namely
\begin{equation}
a_i \gamma_i + a_j \gamma_j + a_h \gamma_h + a_k \gamma_k + a_{s - 1} \gamma_{s - 1}
\label{eq:cond}
\end{equation}
where each two successive indexes are distinct, i.e. $i\neq j\neq h \neq k \neq s-1$. According to \eqref{eq:SMCexpomatrix}, the coefficients $a_l$'s include only the elements $p_{i0} = 0$ and $p_{i1}$, $i = 0, \ldots, m - 1$. Hence, $a_0= 0$, while, if present in \eqref{eq:cond},  $\gamma_1 = 1$. 
Having assumed that the Tanner graph relative to the submatrix $\left[\begin{array}{@{}c@{}|@{}c@{}} \vec{0} & \vec{P}_{1} \end{array}\right]$ has no strictly avoidable cycles with length up to $10$, it has to be $a_i\neq 0$, $\forall i \in [1,\ldots, s-1]$. In order to ensure that the whole expression is different from $0$ as well, it is sufficient to choose
\begin{equation}
\gamma_{s - 1} > \left|\frac{- a_i \gamma_i - a_j \gamma_j - a_h \gamma_h - a_k \gamma_k}{a_{s - 1}}\right|
\end{equation}
with values $\in \{\gamma_{s - 2} + 1, \ldots, N - 1 \}$.
This condition must hold for any element of $\mathcal{I}_{0,1,\ldots,s - 1}^{k}$. So, setting $\lambda_{0,1, \ldots, s - 1}^{k}=\max\{\vert x\vert\;|\;x\in \mathcal{I}_{0,1, \ldots, s - 1}^{k}\}$, and $N > \lambda_{0,1, \ldots, s - 1}^{k}$, all the elements in $\mathcal{I}_{0,1,\ldots,s - 1}^{k}$ are non-zero mod $N$. The final value of $N$ results at the end of this analysis, that is for $s = n$.
\end{proof}
\begin{Example} Let $m=3$ and $n=6$. Suppose that $\mathbf{P}^{\mathrm{SMC}}_{3\times 6}$ is the exponent matrix of a QC-LDPC block code $C$, as defined in (\ref{eq:SMCexpomatrix}), such that $\vec{P}_{1}=\left( 0,1,29 \right)^T$. Considering (\ref{fore}), it is easy to check that the Tanner graph associated to $\left[\begin{array}{@{}c@{}|@{}c@{}} \vec{0} & \vec{P}_{1} \end{array}\right]$ contains no strictly avoidable cycles of length up to $10$. Then, according to Proposition \ref{SMCproposition}, the Tanner graph of $C$ has no strictly avoidable cycle of length up to $10$ for sufficiently large $N$ and a proper choice of $\gamma_j$'s. Choosing $\gamma_{2}=3$, $\gamma_{3}=7$, $\gamma_{4}=67$ and $\gamma_{5}=144$ and $N=271$, it is easily verified that $C$ has $g=12$. The code length is $L = 1626$.
\end{Example}

\section{Greedy search algorithm}
\label{Sec4}
Based on Proposition \ref{SMCproposition}, we have developed a search algorithm that finds the smallest possible $\gamma_j\in \{\gamma_{j - 1} + 1, \ldots, N - 1\}$, $j = 2,\ldots,n - 1$, that leads to $g=12$. From the complexity viewpoint, we can estimate the advantage resulting from the SMC assumption by considering that the exhaustive search of $\mathbf{P}$ requires to find $mn$ elements. Instead, with our method, we only need to find $m + n - 4$ values, namely, $m - 2$ entries of the vector $\vec{P}_{1}$ and $n - 2$ multiplication factors.
A formal description of the proposed greedy search procedure is given in Algorithm \ref{Algo1}. As inputs, it takes $m,\;n,\;N$ $(m<n\leq N)$, with $m,n,N\in \mathbb{N}$, $k$ ($= 2, 3, 4, 5$), and an all zero matrix $\mathbf{P}$ of dimension $m\times n$. As output, it returns $0$ if there is no feasible solution, or an exponent matrix with girth $g \geq 2k$, otherwise. Moreover, the minimum possible $m_{h}$ for each exponent matrix of SC-LDPC-CCs has been found through the min-max optimization model, recently proposed in \cite{MBAT2017}.

\begin{algorithm} 
\caption{Greedy search algorithm}\label{Algo1} 
\small \hspace*{\algorithmicindent} \textbf{Input:}\;$m,\;n,\;N,\;k$ and zero matrix $\mathbf{P}=\left[\begin{array}{@{}c@{}|@{}c@{}|@{}c@{}|@{}c@{}|@{}c@{}}
\vec{0} & \vec{P}_{1} & \vec{P}_{2} & \ldots & \vec{P}_{n-1}
\end{array}\right]$\\
\hspace*{\algorithmicindent} \textbf{Output:}\; Exponent matrix $\mathbf{P}$ with girth $2k$
\begin{algorithmic}[1] 
\State $p_{10} \gets 0,\;p_{11} \gets 1,\;\mathcal{S}_2 \gets (m-2)\text{-combinations of}\;\{2,\ldots,N-1\}$
\BState \emph{top}: 
\State Pick $(p_{21},\ldots,p_{(m-1)1})$ from $\mathcal{S}_2$ in a way that $p_{21}<p_{31}<\ldots<p_{(m-1)1}$
\State $\mathcal{S}_2 \gets \mathcal{S}_2\setminus \{(p_{21},\ldots,p_{(m-1)1})\}$
\If {at least one of the relations in $\mathcal{I}_{0,1}^{k-1}$ results in a cycle and $\vert \mathcal{S}_2\vert >0$}
\State \textbf{goto} \emph{top}
\ElsIf {$\vert \mathcal{S}_2\vert =0$}
\State \textbf{return $0$}
\EndIf 
\For {$j:\;2\;\text{to}\;n-1$}
\State $\mathcal{N}_{j} \gets \{\gamma_{j-1}+1,\ldots,N-1\}$
\BState \emph{loop}: 
\State Pick $\gamma_{j}$ from the set $\mathcal{N}_{j}$
\State $\mathcal{N}_{j}\gets \mathcal{N}_{j}\setminus \{\gamma_{j}\}$
\State $\vec{P}_{j} \gets \gamma_{j} \otimes \vec{P}_{1}$
\If {at least one of the relations in $\mathcal{I}_{0,1,\ldots,j}^{k-1}$ results in a cycle and $\vert \mathcal{N}_{j}\vert >0$} 
\State \textbf{goto} \emph{loop}
\ElsIf {$\vert \mathcal{N}_{j}\vert =0$} 
\State \textbf{return $0$}
\EndIf
\EndFor
\State \textbf{return $\mathbf{P}$}
\end{algorithmic}
\end{algorithm}

\section{Numerical Results}
\label{Sec5}

By applying the method presented in the previous sections, we have designed several codes with values of $N$  and $m_h$ in many cases significantly smaller than those of classical codes with the same code rate and girth. The method that in Section \ref{Sec3} has been illustrated for the case of $g = 12$ has been applied also for the case of $g = 10$.
In particular, we have considered $m = 3, 4$ and $n = 4, \ldots, 12$ for the QC-LDPC block codes, and $c = 3, 4$ and $a = 4, \ldots, 12$ for the SC-LDPC-CCs. The obtained values of $N$ and $m_h$ have been compared with those resulting from the application of classical design approaches reported in \cite{sullivan1,MBAT2017,Bocharova1,ATasdighi1,Amirzadeh1}, which,
to the best of our knowledge, are those producing the codes with the minimum values of $N$ and $m_h$.
The comparison with our results is shown in Figs. \ref{fig:LiftDegVSn} and \ref{fig:SynForMeOrdVSa}.
We see that the values obtained through our approach are everywhere smaller (often significantly) than those derived with the previous solutions. Let us denote as $\tilde{N}$ ($\tilde{m}_h$) the smallest lifting degree (syndrome former memory order) found with our approach, and as $N^*$ ($m_h^*$) the minimum value found through previous approaches. The ratio of the decoding latency of the newly designed QC-LDPC block codes over that of the classical ones is
\begin{equation}
\Theta_N=\frac{\tilde{N}}{N^*}
\label{eq:ratioQC}.
\end{equation}
Based on \eqref{cas}, we can also assess the ratio of decoding complexity (per output bit) and latency achieved by the newly designed SC-LDPC-CCs over the classical ones, as
\begin{equation}
\Theta_{m_h}=\frac{\tilde{m}_h+1}{m_h^*+1}.
\label{eq:ratios}
\end{equation} 
The values of $\Theta_N$ and $\Theta_{m_h}$ should be kept as small as possible if we aim at minimizing the decoding latency and complexity. We have obtained values of $\Theta_N$ as small as $0.47$, which means a reduction in decoding latency by more than $50\%$, and values of $\Theta_{m_h}$ as small as $0.23$, yielding a reduction of $\Lambda_{\mathrm{SW}}$ and $\Gamma_{\mathrm{SW}}$ by more than $75\%$, with respect to previous solutions.

As a further benchmark of the newly designed codes, we have estimated the bit error rate (BER) of our SC-LDPC-CCs through Monte Carlo simulations of binary phase shift keying modulated transmissions over the additive white Gaussian noise channel, and compared it with that of some codes constructed following \cite{MBAT2017,Bocharova1}. A full-size BP decoder and a BP-based SW decoder, both performing $100$ iterations, have been used in the simulations. Decoding is performed on a full codeword of length $L\rightarrow \infty$ (practically, $L\approx 6\times 10^4$) when the full BP decoder is considered, whereas the SW decoder works over sliding windows of $Wa$ bits each. Let us consider two of our codes ($C_1$ and $C_2$) with $R=\frac{4}{7}$ and $g=10$ and $12$, two previous codes ($C_{B1}$ and $C_{B2}$) designed following \cite{MBAT2017} and two codes ($C_{B3}$ and $C_{B4}$) obtained by unwrapping QC-LDPC block codes designed as in \cite{Bocharova1}. The parameters of these codes are summarized in Table \ref{table:Tabparamsc}. Their BER performance is shown in Fig. \ref{fig:performance10}. We notice that the performance degradation is minimal for both very large and small window sizes. So, we can conclude that the new SC-LDPC-CCs do not exhibit any significant loss with respect to the classical codes, while they enjoy reduced latency and complexity. The value of $\Theta_{m_h}$, according to \eqref{eq:ratios}, is also shown in Table \ref{table:Tabparamsc}.

\begin{figure}
\centering
\includegraphics[width=78mm,keepaspectratio]{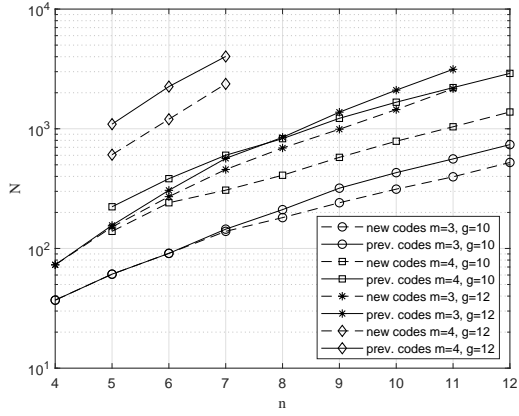}
\caption{Minimum lifting degree ($N$) of new and previously designed QC-LDPC codes versus $n$, for $m=3,4$.}
\label{fig:LiftDegVSn}
\end{figure}

\begin{figure}
\centering
\includegraphics[width=78mm,keepaspectratio]{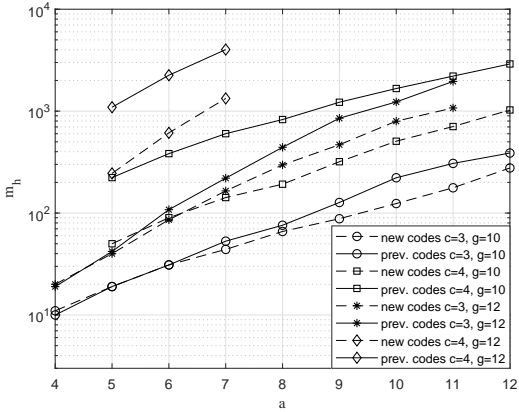}
\caption{Minimum syndrome former memory order ($m_h$) of new and previously designed SC-LDPC-CCs versus $a$, for $c=3,4$.}
\label{fig:SynForMeOrdVSa}
\end{figure}

\begin{table}[!t]
\renewcommand{\arraystretch}{1}
\caption{Values of $a$, $c$, $m_h$, $v_s$ and $g$ of the considered SC-LDPC-CCs with $R=\frac{4}{7}$.}
\label{table:Tabparamsc}
\centering
\begin{tabular}{|c|c|c|c|c|c|c|c|c|}
\hline
Code & $a$ & $c$ & $m_h$ & $v_s$ & $g$&$\Theta_{m_h}$\\ \hline\hline
$C_1$ & $7$ & $3$ & $44$ & $315$ & $10$ & $-$\\ \hline
$C_{2}$ & $7$ & $3$& $165$ & $1162$ & $12$&$-$\\ \hline
$C_{B1}$ & $7$ & $3$ & $53$ & $378$ & $10$&$0.83$\\ \hline
$C_{B2}$ & $7$ & $3$ & $220$ & $1547$ & $12$&$0.75$\\ \hline
$C_{B3}$ & $7$ & $3$ & $88$ & $623$ & $10$&$0.51$\\ \hline
$C_{B4}$ & $7$ & $3$ & $432$ & $3031$ & $12$&$0.38$\\ \hline
\end{tabular}
\end{table}

\begin{figure}[t!]
\centering
\subfloat[][]{\includegraphics[width=0.25\textwidth]{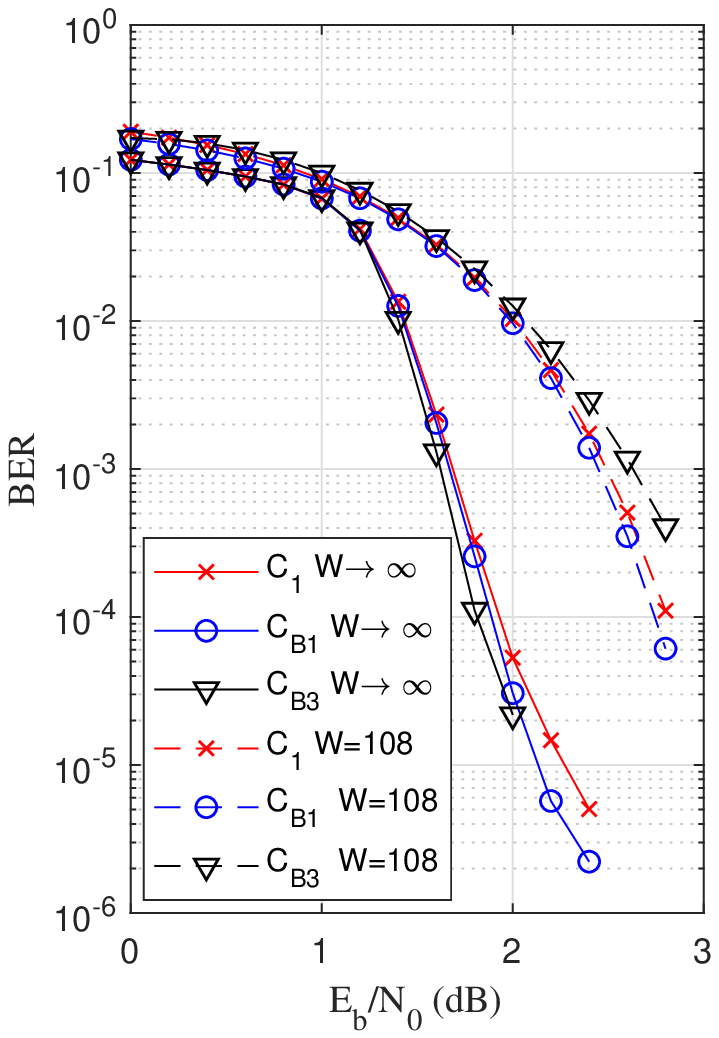}\label{fig:samelay}}
\subfloat[][]{\includegraphics[width=0.25\textwidth]{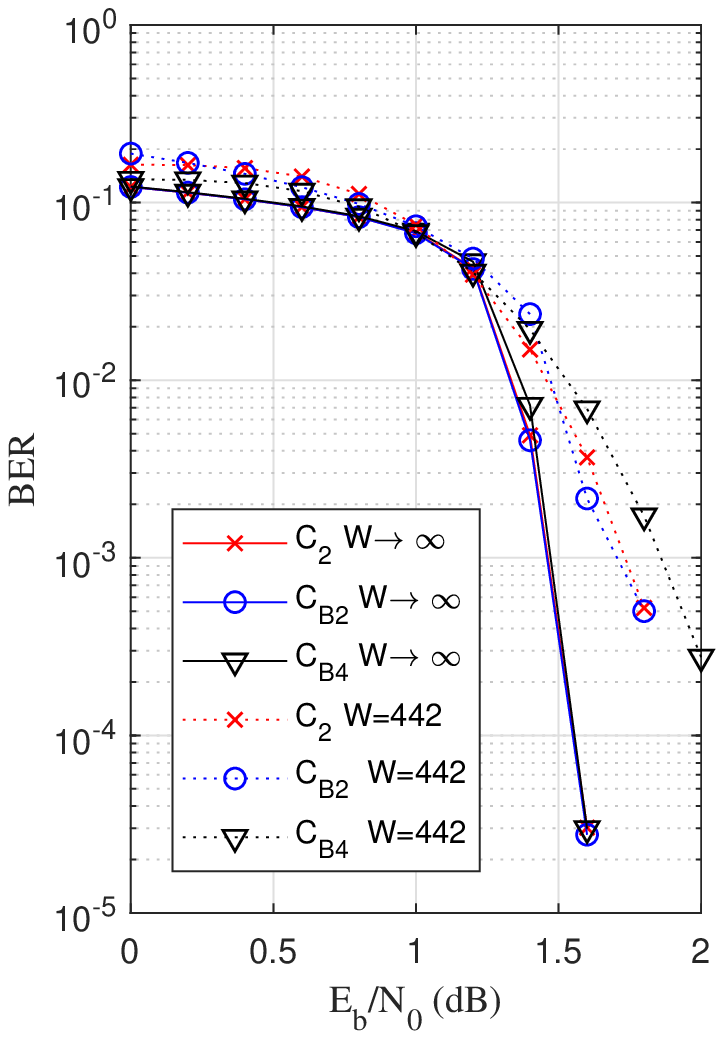}\label{fig:difflay}}
 \caption{Simulated performance of SC-LDPC-CCs with: (a) $g=10$ and (b) $g=12$ as a function of the signal-to-noise ratio.}\label{fig:performance10}
\end{figure}

\section{Conclusion}
\label{Sec6}

We have proposed a method for the design of QC-LDPC block codes able to achieve girths $g = 10, 12$ with very short block lengths. The same approach has been used to design SC-LDPC-CCs with very short syndrome former constraint length.
Our method achieves important reductions in the decoding latency and per output bit complexity at the cost of negligible performance losses over classical approaches.


\end{document}